\documentclass[12pt]{article}
\textheight      = \paperheight
\advance\textheight by -2truein
\textwidth       = \paperwidth
\advance\textwidth  by -2truein
\usepackage{amsmath,amsthm,amsfonts,amssymb,amsopn,amscd}     

\usepackage{color}

\newcommand{\utilde}[1]{\underset{\sim}{#1}}

\def\tilde{\widetilde}

\def\A{{\cal A}}
\def\B{{\cal B}}

\def\D{{\cal D}}

\def\F{{\cal F}}

\def\M{{\cal M}}
\def\N{{\cal N}}
\def\R{{\cal R}}

\def\H{{\cal H}}
\def\K{{\cal K}}
\def\S{{\cal S}}

\def\tcP{{\tilde\cP_+}}

\def\PSL{{{\rm PSL}(2,\mathbb R)}}

\def\S2{S^{1(2)}}
\def\Poi{{\mathcal P}_+^\uparrow}
\def\tPoi{\tilde{\mathcal P}_+^\uparrow}

\newtheorem{theorem}{Theorem}[section]
\newtheorem{lemma}[theorem]{Lemma}

\newtheorem{corollary}[theorem]{Corollary}

\newtheorem{proposition}[theorem]{Proposition}

\theoremstyle{definition} 
\newtheorem{definition}[theorem]{Definition}

\theoremstyle{remark} 
\newtheorem{remark}[theorem]{Remark}

\def\PSL{PSU(1,1)}

\def\RR{{\mathbb R}}
\def\CC{{\mathbb C}}
\def\NN{{\mathbb N}}
\def\ZZ{{\mathbb Z}}

\def\SL2{{{\rm SL}(2,\RR)}}
\def\SLC{{{\rm SL}(2,\CC)}}
\def\PSL2{{{\rm PSL}(2,\RR)}}

\def\U1{{{\rm V}(1)}}
\def\SU2{{{\rm SV}(2)}}

\def\SU{{{\rm SU}}}

\def\A{{\mathcal A}}
\def\B{{\mathcal B}}

\def\D{{\mathcal D}}
\def\F{{\mathcal F}}
\def\H{{\mathcal H}}

\def\cL{\mathcal L}

\def\K{{\mathcal K}}

\def\M{{\mathcal M}}
\def\N{{\mathcal N}}

\def\cO{{\mathcal O}}
\def\P{{\mathcal P}}
\def\cP{{\mathcal P}}
\def\R{{\mathbb R}}
\def\T{{\mathcal T}}
\def\cU{{\mathcal U}}

\def\W{{\mathcal W}}
\def\cW{{\mathcal W}}

\def\eins{{\mathbf 1}}

\begin{document}
\title{The Bisognano-Wichmann property\\ on nets of standard subspaces,\\ some sufficient conditions}
\author{
{\bf Vincenzo Morinelli}\footnote{Supported by the ERC advanced grant 669240 QUEST ``Quantum Algebraic Structures and Models", and INdAM.}\\Dipartimento di Matematica, Universit\`a di Roma Tor Vergata\\
Via della Ricerca Scientifica, 1, I-00133 Roma, Italy\\
email: {\tt morinell@mat.uniroma2.it},}
%
\date{}
\maketitle
\begin{abstract}
We discuss the Bisognano-Wichmann property for local Poincar\'e covariant nets of standard subspaces.  We provide a sufficient algebraic condition on the covariant representation ensuring  the Bisognano-Wichmann and the Duality properties  without further assumptions on the net. We call it {modularity} condition. It holds for direct integrals of scalar massive and massless representations. We present a class of massive modular covariant nets not satisfying the Bisognano-Wichmann property. Furthermore, we give an outlook on the relation between the Bisognano-Wichmann property and the Split property in the standard subspace setting.
\end{abstract}

\section{Introduction}
The Algebraic Quantum Field Theory (AQFT) is a very fruitful approach to study properties of quantum fields by using operator algebras. Models of local quantum field theories are described by nets of von Neumann algebras on a fixed spacetime, satisfying basic relativistic and quantum assumptions. 
A local Poincar\'e covariant net of von Neumann algebras on a fixed Hilbert space $\H$ is a local isotonic map $\K\ni\cO\mapsto \A(\cO)\subset \B(\H)$ from the set $\K$ of open, connected bounded regions of the $1+3$ dimensional Minkowski spacetime $\mathbb R^{1+3}$ to von Neumann algebras $\A(\cO)\subset\B(\H)$. It is assumed the existence of  a unitary positive energy Poincar\'e representation  $U$ on $\H$ acting covariantly on $\A$ and of a unique (up to a phase) normalized $U$-invariant vector $\Omega\in\H$  which is cyclic for all the local algebras, namely the vacuum vector. The deep connection between the algebraic structure and the geometry of the model is a very fascinating fact. 
In \cite{BW1,bw2} Bisognano and Wichmann showed that any model coming from Wightman fields encloses within itself the information on its geometry.
The authors proved that the  modular operators related to  algebras associated with wedge-shaped regions and the vacuum state have a geometrical meaning: they implement pure Lorentz transformations. This is expressed through the {\bf Bisognano-Wichmann property (B-W)}:
\begin{equation}
\label{eq:biwi}
U(\Lambda_W(2\pi t))=\Delta_{\A(W),\Omega}^{-it}\end{equation} for any wedge  region $W$, where $t\mapsto\Lambda_{W}(t)$ is the one-parameter group of boosts associated with the wedge $W$, and $\Delta^{it}_{\A(W),\Omega}$ is the modular group of the von Neumann algebra $\A(W)$  generated by $\cup_{\cO\subset W}\A(\cO)$ w.r.t. the vacuum vector $\Omega$.

 It can be stated something more.  A von Neumann algebra net $\A$ is said to be \textbf{modular covariant} if 
\begin{equation}\label{eq:modco}
\Delta^{it}_{\A(W),\Omega}\A(O)\Delta^{-it}_{\A(W),\Omega}=\A(\Lambda_W(-2\pi t)O), \quad \text{for any } O\in\K
\end{equation} 
i.e. the modular group associated with any wedge algebra implements the  covariant action of the associated one-parameter group of boosts on the local net. 
The modular covariance property is introduced in \cite{BGLco}, it is weaker than the B-W property 
and it ensures that there exists a covariant unitary positive energy representation of the Poincar\'e group generated by the modular theory of the net algebras \cite{GL}.  This marked one of the successes of the Tomita-Takesaki theory: once the algebra of the observables and the vacuum state are specified, the modular structure is determined and it has a geometrical meaning.  Sufficient conditions on these properties are given  in \cite{bor,B,BDFS,sch}.

The B-W property cares only about the modular theory of the algebraic model, which is contained in the  real  structure of the net. This can be described by using real standard subspaces of the Hilbert state space (cf. definition in Sect. \ref{sect:stsb}). It is possible to characterize the standard subspaces with an analogue of the Tomita-Takesaki modular theory 
which coincides with the von Neumann algebra Tomita-Takesaki theory  when one considers von Neumann algebras $A\subset\B(\H)$ with a cyclic and separating vector $\Omega\in\H$ and the subspaces $H=\overline{A_{sa}\Omega}\subset\H$.  

At this point it is natural to consider analogous nets of standard subspaces, which provide a very fruitful approach to QFT. For instance, they have a key role in finding localization properties of Infinite Spin particles, cf. Ref. \cite{LMR} or finding new models in low dimensional quantum field theory \cite{LL,LW}. A geometrical approach to nets of standard subspaces is in \cite{NO}.

The B-W property is essential to give a canonical structure to the first quantization nets, hence to the free fields (cf. \cite{BGL}): given a particle, namely an irreducible,  positive energy, unitary representation of the Poincar\'e group, there is a canonical way to build up the associated one-particle net of localized states and its second quantization free field.  We make an analysis in the converse sense. Are those nets in some sense unique? At this level, the question is not well posed since there is some freeness in choosing the modular conjugation implementing the position and time reflection. We can ask  when a (unitary, positive energy) Poincar\'e representation $U$ is {\bf modular, i.e. any net of standard subspaces it acts covariantly on satisfies the B-W property} (cf. Definition \ref{def:mod}). In particular $U$ would be implemented by the modular operators. This would give an answer to the necessity to assume the B-W property instead of deducing it by the assumptions.

An approach to this problem is to show the B-W property by exploiting the analyticity property of the wave functions as in \cite{BE,M}.
 It  is difficult to extend this analytic approach to more general representations as infinite direct sums, direct integrals or to the massless case. 

We are going to present a purely algebraic argument giving a sufficient condition for the modularity of a large family of Poincar\'e representations as direct integrals of scalar representations, including the massless case. The idea comes from the following facts: the Lorentz group acts geometrically on the momentum space and one can check that on the mass shell it is possible to pointwise reconstruct the action of a Lorentz transformation, sending a wedge $W$ to its causal complement $W'$ just considering $W$-fixing transformations (see Remark \ref{rmk:orb}). With this hint we introduce the following condition on a (unitary, positive energy) representation $U$ of the Poincar\'e group, called \textbf{modularity condition} (cf. Definition \ref{def:mod}). We ask that the von Neumann algebra, generated by translations and the  Lorentz subgroup of transformations fixing  $W$ is large enough to enclose the transformations sending $W$ onto $W'$.  In Theorem \ref{thm:starcor} we show that the modularity condition is sufficient to prove that $U$ is modular. 
In particular, for any $U$-covariant standard subspace net  $H$ the quotient between the two sides of \eqref{eq:biwi}  is the identity. This condition does neither depend on the net nor on multiplicity of the representation and passes to some direct integrals. Our analysis holds for scalar Poincar\'e representations in any spacetime $\mathbb R^{1+s}$, with $s\geq3$, see Remark \ref{rmk:dim}. A comment on the massless finite helicity case is given in Remark \ref{rmk:hel}.
We rely on the idea that the modular covariance  has to be a natural assumption in every Quantum Field Theory satisfying basic relativistic and quantum hypotheses.

Known counterexamples to modular covariance seem very artificial as they imply a breakdown of Poincar\'e covariance, see \cite{yng, GY}. In Sect. \ref{sect:infcoun}, we give an explicit example of a massive Poincar\'e covariant standard subspace net,  which is modular covariant  not satisfying the B-W property. The massless case was treated in \cite{LMR}. These kinds of general counterexamples clarify what kind of settings may prevent the identification of the covariant representation with the modular symmetries.

The relation between the split and the modular covariance properties is an interesting problem. In \cite{lodo}, Doplicher and Longo proved that if the dual net of local von Neumann algebras
associated with the scalar generalized free field with K\"allen-Lehmann measure $\mu$ has the split property, then $\mu$ is purely atomic and concentrated on isolated
points. In Sect.  \ref{sect:CO}, we give an outlook in the standard subspace setting of the  problem and of this result. 

The paper is organized as follows. In Sect. 2 we introduce the one-particle setting and we recall relevant results on the real subspace structure, on Poincar\'e representations and on nets of standard subspaces. 
In Sect. 3 we present our algebraic condition and its properties. We prove it for general scalar representations in Sect. 4. In Sect. 5 we show the massive counterexample to the B-W  property. In Sect.6 it is discussed the relation between the B-W and  the Split properties in this one-particle setting.

This paper was reviewed in \cite{MMG}.

\section{One-particle net}

\subsection{Standard subspaces}\label{sect:stsb}
 Firstly, we recall some definitions and basic results on standard
subspaces and (generalized) one-particle models.
A linear, real, closed subspace $H$ of a complex Hilbert space $\H$ is called {\bf cyclic} if 
$H+iH$ is dense in $\H$, {\bf separating} if $H\cap iH=\{0\}$ and 
{\bf standard} if it is cyclic and separating.

Given a standard subspace $H$  the associated {\bf Tomita operator} $S_H$ is defined to be the closed the anti-linear involution with domain $H+iH$, given by: $$S_H:H+iH\ni \xi + i\eta \mapsto \xi - i\eta\in H+iH, \qquad\xi,\eta\in H,$$ on the dense domain $H + iH\subset\H$. The polar decomposition $$S_H = J_H\Delta_H^{1/2}$$ defines the positive self-adjoint {\bf modular operator} $\Delta_H$ and the anti-unitary
{\bf modular conjugation} $J_H$. In particular, $\Delta_H$ is invertible and $$J_H\Delta_H J_H=\Delta_H^{-1}.$$

If $H$ is a  real linear subspace of $\H$, the \emph{symplectic complement} of $H$ is defined by
\[
H' \equiv \{\xi\in\H\ ;\ \Im(\xi,\eta)=0, \forall \eta\in H\} = (iH)^{\bot_\R}\ ,
\]
where $\bot_\R$ denotes the orthogonal in $\H$ viewed as a real Hilbert space with respect to the real part of the scalar product on $\H$.
$H'$ is a closed, real linear subspace of $\H$. If $H$ is standard, then $H = H''$. 
It is a fact that $H$ is cyclic (resp. separating) iff $H'$ is separating (resp. cyclic), thus $H$ is standard iff $H'$ is standard and we have
\[
S_{H'} = S^*_H \ .
\]
Fundamental properties of the modular operator and conjugation are
\begin{equation*}
\Delta_H^{it}H = H, \quad J_H H = H' \ ,\qquad  t\in\RR\ .
\end{equation*}
We shall call the one-parameter, strongly continuous group $t\mapsto \Delta_H^{it}$, the {\bf modular group} of $H$ (cf. \cite{RV}).%

 There is a 1-1 correspondence between Tomita operators and  standard subspaces.
\begin{proposition}\label{prop:11}{\rm \cite{L,LN}.}
The map 
\begin{equation}\label{SH}H\longmapsto S_H
\end{equation}
is a bijection between the set of standard subspaces of $\H$ and the set of closed, densely defined, anti-linear involutions on $\H$. The inverse of the map \eqref{SH} is
$$S\longmapsto\ker({\bf1}-S).$$
Furthermore, this map is order-preserving, namely
$$H_1\subset H_2 \quad\Leftrightarrow\quad S_{H_1}\subset S_{H_2},$$
and we have
$S_{H}^*=S_{H'}.$
\end{proposition}
As a consequence,
$$\left\{\begin{array}{c}
\text{Standard}\\\text{subspaces }\\ H\subset\H\end{array}\right\}\stackrel{1:1}\longleftrightarrow
\left\{\begin{array}{c} \text{closed, dens. def.}\\ \text{ anti-linear }\\ \text{ involutions S}\end{array}\right\} \stackrel{1:1}\longleftrightarrow
\left\{ \begin{array}{c} (J,\Delta) \text{ pairs of an}\\ \text{anti-unitary}\\ \text{ involution  and  a}\\ \text{positive self-adjoint}\\ \text{operator on}\, \H \text{ s.t. }\\ J\Delta J=\Delta^{-1}
\end{array}\right\}
$$
Here is a basic result on the standard subspace modular theory.
\begin{lemma}\label{lem:sym}{\rm \cite{L,LN}.}
Let $H,K\subset\H$ be standard subspaces  and $U\in\cU(\H)$ be a unitary operator on $\H$ such that $UH=K$. Then $U\Delta_H U^*=\Delta_K$ and $UJ_HU^*=J_K$.
\end{lemma}
The following is the analogue of the Takesaki theorem for standard subspaces.
\begin{lemma}\label{inc}{\rm \cite{L,LN}.}
Let $H\subset \H$ be a standard subspace, and $K\subset H$ be  a closed, real linear subspace of $H$. 

If $\Delta_H^{it}K=K$, $\forall t\in\R$, then $K$ is a standard subspace of $\K\equiv \overline{K+iK}$ and $\Delta_H|_K$ is the modular operator of $K$ on $\K$. 
Moreover, if $K$ is a cyclic subspace of $\H$, then $H=K$.
\end{lemma}

The following is the Borchers theorem in the standard subspace setting.
\begin{theorem}{\rm \cite{L,LN}.}\label{Borch}
Let $H\subset\H$ be a standard subspace, and $U(t)=e^{itP}$ be a one-parameter unitary group on $\H$ with  generator $\pm P>0$, such that $U(t)H\subset H$, $\forall t\geq 0$. Then, \begin{equation}\label{eq:tradil}\left\{\begin{array}{ll}
\Delta_H^{is}U(t)\Delta_H^{-is}= U(e^{\mp2\pi s}t)&\\J_HU(t)J_H=U(-t)\end{array}\right.\qquad \forall t,s\in\RR.\end{equation}
\end{theorem}
In other words, the last theorem claims that if there is a one-parameter unitary group $t\mapsto U(t)$, with a positive generator, which properly translates a standard subspace, then \eqref{eq:tradil} establishes the unique (up to multiplicity), positive energy representation of the translation-dilation group.

The converse of the Borchers theorem can be stated in the following way.
\begin{theorem}\label{thm:boconv}{\rm \cite{BGL}.}
Let $H$ be a standard space in the Hilbert space $\H$ and $U(t) = e^{i t P}$
a one-parameter group of unitaries on $\H$ satisfying:
$$\Delta_H^{is}U(t)\Delta_H^{-is}=U(e^{\mp2\pi s}t)\quad\text{and}\quad J_HU(t)J_H=U(-t),\qquad\forall t,s\in\R$$

The following are equivalent:
\begin{enumerate} 
\item $U(t)H \subset H$ for $ t\geq 0$;
\item $\pm P$ is positive.
\end{enumerate}
\end{theorem}

\subsection{The Minkowski space and the Poincar\'e group}\label{sect:MP}
Let $\R^{1+3}$ be the Minkowski space, i.e. a four-dimensional real manifold endowed with the metric $$(x,y)=x_0y_0-\sum_{i=1}^3x_iy_i.$$
In a 4-vector  $x=(x_0,x_1,x_2,x_3)$   $x_0$ and $\{x_i\}_{i=1,2,3}$ are the time and space coordinates, respectively. The Minkowski space has a causal structure induced by the metric. The causal complement of a region $O$ is given by $O'=\{x\in\RR^{1+3}: (x-y)^2<0, \forall y\in O \},$ where $(x-y)^2=(x-y,x-y)$ refers to the norm induced by the metric. A {\bf causally closed} region is such that $O=O''$.

We shall denote with $\cP$ the {\bf Poincar\'e group}, i.e the inhomogeneous symmetry group of $\RR^{1+3}$. It is the semidirect product of the Lorentz group $\cL$, the homogeneous Minkowski symmetry group, and the $\R^4$-translation group , i.e. $\cP=\R^4\rtimes\cL$. It has four connected components, as $\cL$ has four connected components, and we shall indicate with $\Poi=\R^4\rtimes\cL_+^\uparrow$ the connected component of the identity. One usually  refers to $\Poi$ as the {\bf proper ortochronous} (connected component of the) {\bf Poincar\'e group}.  $\mathcal L_+^\uparrow$  is not simply connected. The $\mathcal{L}_+^\uparrow$ universal covering $\tilde{\mathcal{L}}_+^\uparrow$ is $\SLC$, thus the $\Poi$-universal covering $\tPoi$ is isomorphic to $\RR^4\rtimes \SLC$. Let $\Lambda:\R^4\rtimes\SLC\ni (a,A)\longmapsto\Lambda(a,A)\in \Poi$ be the \textbf{covering map}.

Unitary positive energy representations of the (universal covering of the) Poincar\'e group belong to three families: massive, massless finite helicity and massless infinite spin. Massive representations are labelled by a mass parameter $m\in (0,+\infty)$ and a spin parameter $s\in\frac\NN2$; massless finite helicity and infinite spin representations have zero mass and are  labelled by an helicity parameter $h\in\frac\ZZ2$ and a couple $(\kappa,\epsilon)$ where $\kappa\in \R^+$ is the radius and $\epsilon\in\{0,\frac12\}$ the bose/fermi alternative, respectively. 

We shall indicate with $\cP_+$ the subgroup of $\cP$ generated by $\Poi $ and the space and time reflection $\theta$. Consider the automorphism $\alpha$ of the Poincar\'e group $\theta$ generated by the adjoint action of the $\theta$-reflection. The {\bf proper Poincar\'e group} $\cP_+ $ is generated as a semidirect product  $$\Poi\rtimes_\alpha \ZZ_2$$ through the $\alpha$-action. It is  well known (see for example \! \cite{var}) that any irreducible representation of the Poincar\'e group $U$, except for finite helicity, extends to an (anti-)unitary representation of $\tilde U$ of $\tPoi$, i.e.
$$\tilde U(g)\;\text{is}\;\left\{
\begin{array}{ll}
\text{is unitary}&g\in\tPoi \\ 
\text{is anti-unitary}&g\in\tilde\cP_+^\downarrow \dot\,= \,\theta\cdot\Poi
\end{array}\right.$$

At this point it is necessary to fix some notations about regions and isometries of the Minkowski spacetime. A {\bf wedge-shaped region} $W\subset\RR^{1+3}$ is an open region of the form  $gW_1$ where $g\in\Poi$ and $W_1=\{x\in\RR^{1+3}: |x_0|<x_1\}$.    The set of wedges is  denoted by $\W$. Let $\W_0\subset\W$ be the subset of wedges the form $gW_1$ where $g\in \cL_+^\uparrow$. Note that if  $W\in\W$ (or $W\in\W_0$), then $W'\in\W$ (resp. $W'\in\W_0$). 
For every wedge region $W\in\W$ there exists a unique one-parameter group of Poincar\'e boosts  $t\mapsto\Lambda_W(t)$ preserving $W$, i.e. $\Lambda_W(t) W=W$ for every $t\in\RR$. It is defined by the adjoint action of a $g\in\Poi$ such that $gW_1=W$ on $\Lambda_{W_1}$, where $\Lambda_{W_1}(t)x=(x_0\cosh t + x_1\sinh t,x_0\sinh t + x_1\cosh t, x_2,x_3 )$. Let $t\mapsto\lambda_W(t)$ be the (unique, one-parameter group) lift to the covering group.
We shall denote with $W_\alpha\in\W$ the wedge along $x_\alpha$-axis, i.e. $W_\alpha=\{x\in\RR^{1+3}: |x_0|<x_\alpha\}$ with $\alpha=1,2,3$. Let  $t\mapsto \Lambda_\alpha(t)$ and $\theta\mapsto R_\alpha(\theta)$ be respectively the boosts and the rotations of $\cP_+^\uparrow$ fixing $W_\alpha$ and $t\mapsto\lambda_j(t)$ and $\theta\mapsto r_j(\theta)$ be the corresponding lifts to $\SLC$. Note that  $\lambda_j(t)=e^{\sigma_jt/2}$ and $r_j(\theta)=e^{i\sigma_j\theta/2}$,  where $\{\sigma_i\}_{i=1,2,3}$ are the Pauli matrices. 

\subsection{One-particle nets}

Let $U$ be a unitary positive energy representation  of the Poincar\'e group $\tPoi$ on  an Hilbert space $\H$. 
We shall call a \textbf{$U$-covariant (or Poincar\'e covariant) net of standard subspaces on wedges}  a map $$H:\W\ni W\longmapsto H(W)\subset\H,$$
associating to every wedge in $\RR^{1+3}$ a closed real linear subspace of $\H$, satisfying the following properties:
 \begin{enumerate}
\item \textbf{Isotony:}  If $W_1,W_2\in\W$ and $W_1\subset W_2$ then $H(W_1)\subset H(W_2)$;
\item \textbf{Poincar\'e Covariance:}   $U(g)H(W)=H(gW),$ $\forall g\in\tPoi,\,\forall W\in\W$;
\item \textbf{Positivity of the energy:} the joint spectrum of translations in $U$ is contained in the forward light cone $V_+=\{x\in\RR^{1+3}: x^2=(x,x)\geq0\textrm{ and } x_0\geq0\}$
\item \textbf{Reeh-Schlieder property (R-S):} if $W\in \W$, then  $H(W)$ is a cyclic subspace of $\H$;
\item \textbf{Twisted Locality:}  there exists a self-adjoint unitary operator $\Gamma\in U(\tilde\cP_+^\uparrow)'$, s.t. $\Gamma H(W)=H(W)$ for any $ W\in\cW$ and if $W_1\subset W_2'$ then $$ B H(W_1)\subset H(W_2)',$$ with $\displaystyle{ B=\frac{1+i\Gamma}{1+i}}.$\end{enumerate} 

We shall indicate a $U$-covariant net  of standard subspaces on wedges  $W\mapsto H(W)$   satisfying 1.-5.  with the couple $(U,H).$
This is the setting we are going to study the following two  properties:

\begin{enumerate}
\item[6.]\textbf{ Bisognano-Wichmann property}: if $W\in\W$, then \begin{equation}\label{eq:biwitwist}
U(\lambda_W(2\pi t))=\Delta^{-it}_{H(W)},\qquad\forall t\in\RR;\end{equation}
\item[7.]\textbf{ Duality property}: if $W\in\W$, then $H(W)'=BH(W')$.
\end{enumerate}
Clearly, if $U$ factors through $\Poi$ then the two expressions of the B-W property \eqref{eq:biwi} and  \eqref{eq:biwitwist} coincide.

The relations between the modular theory of the wedge subspaces and the twisted operator are expressed by the following proposition.
 \begin{proposition}\label{prop:B} The following hold
$$[\Delta_{H(W)},B]=0$$
$$J_{H(W)}BJ_{H(W)}=B^*$$
\end{proposition}
\begin{proof}
As $\Gamma H(W)=H(W)$ then, by Lemma \ref{lem:sym}, $\Gamma\Delta_{H(W)}\Gamma^*=\Delta_{H(W)}$ and $\Gamma J_{H(W)}\Gamma^*=J_{H(W)}$. A straightforward computation concludes the argument.
\end{proof}
\begin{proposition}\label{ssd}Wedge duality is consequence of the B-W property.\end{proposition} 
\begin{proof}
By Proposition \ref{prop:B} $\Delta_{H(W)}=\Delta_{BH(W)}$ and by covariance $$ H(W')=U(\Lambda_W(-2\pi t))H(W')=\Delta_{H(W)}^{it} H(W').$$  By twisted locality $BH(W')\subset H(W)'$ and Lemma \ref{inc} we get the thesis.
\end{proof}

It is possible to define closed real linear subspaces associated with bounded causally closed regions as follows 
\begin{equation}\label{HO}
H(O)\dot=\bigcap_{\W\ni W\supset O}H(W).
\end{equation}
This defines a net of real subspaces on causally closed regions $O\mapsto H(O)$. 
Note that $H(O)$ is not necessarily cyclic. If $H$ is a net satisfying 1.-7. assumptions and $H(O)$ is cyclic, then
\[
H(W) = \overline{\sum_{O\subset W} H(O)}
\]
by Lemma \ref{inc}. If $H(O)$ is cyclic for any double cone $O$, we say that the net $O\mapsto H(O)$ satisfies the {\bf R-S  property for double cones}.

In \cite{BGL}, Brunetti, Guido and Longo showed that there is a 1-1 correspondence between (anti-)unitary, positive energy representations of $\cP_+$ and covariant nets of standard subspaces satisfying the B-W property. For the sake of completeness, we recall their theorem and we present the proof in the fermionic case which is not contained in the original paper.

\begin{theorem}\label{thm:bgl}{\rm\cite{BGL}.}
There is a 1-1 correspondence between:
\begin{itemize}
\item[a.] (Anti-)unitary positive energy representations of $\tilde\cP_+$.
\item[b.] Local nets of standard subspaces  satisfying 1-7.
\end{itemize}
\end{theorem}
\begin{proof}
Consider the automorphism of the Poincar\'e group $\tcP$ generated by the adjoint action of $j_3$, one of the two $\tcP$ elements implementing the $x_0-x_3$ reflection. One can check in $\tilde\P_+$ that
\begin{equation}\label{eq:1}
j_3(a,A)j_3=(j_3a,\sigma_3A\sigma_3),\qquad \forall (a,A)\in\R^4\rtimes \SLC 
\end{equation} and 
\begin{equation}\label{eq:2}
r_1(\pi)\lambda_{3}(t)r_1(\pi)^{-1}=\lambda_3(-t)\qquad \mathrm{and\qquad }r_1(\pi)j_{3}r_1(\pi)^{-1}=-j_{3},
\end{equation}
(e.g. cf. Appendix in \cite{M}).
Consider an (irreducible) fermionic representation $U$ of $\tilde \cP_+$, namely a $\tilde\cP_+$-representation  which does not factor through $\cP_+$. In particular $U(2\pi)=-{\bf1}$. Since $U$ lifts to a representation of the Lie algebra of $\tPoi$ on the G\"arding domain and by relations \eqref{eq:1} and \eqref{eq:2} we get
\begin{equation}\label{eq:bgl1}U(R_1(\pi))K_{3}U(R_1(\pi))^*=-K_{3},\end{equation}
\begin{equation}\label{eq:bgl2}U(R_1(\pi))J_{3}U(R_1(\pi))^*=-J_{3},\end{equation}
where $U(\lambda_3(t))=e^{iK_3t}$ and  $J_3$ denotes $U(j_3)$ (choose one of the two possible choices for $U(j_3)$ in $\tcP$). The anti-unitary operator $J_{3}=U(j_3)$  and the self-adjoint operator $\Delta_{W_3}=e^{2\pi K_3}$ satisfy $J_3\Delta_{W_3}J_3=\Delta^{-1}_{W_3}$. Hence, it is possible to define an anti-unitary involution $S_{W_3}=J_{3}\Delta_{W_3}^{1/2}$, and, by Proposition \ref{prop:11},  a subspace $H(W_3)$ associated with the $W_3$ wedge.  Clearly $S_{W_3}$ is the Tomita operator $S_{H(W_3)}$ of the subspace $H(W_3)$. By covariance, a map of standard subspaces $\W\ni W\longmapsto H(W)\subset{\H}$ is well defined. Indeed, for any wedge $W$, $S_{H(W)}$ is the Tomita operator determining $H(W)$, defined as $S_{H(W)}=U(g)S_{H(W_3)}U(g)^*$, with $g\in\P_+^\uparrow$ such that $gW_3=W$. Note that $S_{H(W)}=S_{H(r(2\pi)W)}$ and $S_{H(W)}$ is well defined. This clarify the ambiguity in the choice of $J_3$. Furthermore, by covariance and relations \eqref{eq:bgl1} and \eqref{eq:bgl2}, as $S_{H(W)}=J_W\Delta_W^{1/2}$ then $S_{H(W')}=-J_W\Delta_{H(W)}^{-1/2}$. 
It easily follows that $H(W')=iH(W)'.$ This ensures twisted locality and duality as we can define $\Gamma\dot=U(2\pi)\in U(\tPoi)'$ and $B=\frac{1-i}{1+i}\cdot{\bf1}=-i\cdot{\bf1}$ is the twist operator, i.e. $H(W')=BH(W)'$ where $W\in\W.$

 Positivity of the energy, Poincar\'e covariance, B-W and R-S properties are ensured  by construction. Isotony follows as in \cite{BGL} by positivity of the energy and Theorem \ref{thm:boconv}. \end{proof}
%

\section{A modularity condition for the Bisognano-Wichmann property}

We define the following subgroups of $\tPoi$:
\begin{itemize}
\item $G_W^0\dot=\{A\in\SLC: \Lambda(A) W=W\}$, where $W\in\W_0$. It is the subgroup of $\tilde\cL_+^\uparrow$ elements fixing $W$ through the covering homomorphism $\Lambda$. 
\item  $G_W =\langle G_W^0,\T \rangle$,  with $W\in\W_0$, where $\T$ is the $\RR^{1+3}$-translation group. $G_W$ is the group generated by $G_W^0$ and $\T$.
\item For a general wedge  $W\in\W$, $G_W^0$ and $G_W$  are defined  by the transitive action of $\Poi$ on wedges.
\end{itemize}  

Let $W\in\W$. Consider the strongly continuous map 
\begin{equation}\label{eq:Z}Z_{H(W)}:\RR\ni t\mapsto \Delta_{H(W)}^{it}U(\lambda_W(2\pi t)).\end{equation} 
It has to be the identity map if the B-W property \eqref{eq:biwitwist} holds.

\begin{proposition}\label{prop:sym} 
Let $(U,H)$ be a Poincar\'e covariant net of standard subspaces. Then, for every $W\in\W$, the map $t\mapsto Z_{H(W)}(t)$
defines a one-parameter group and
 $$Z_{H(W)}(t)\in U(G_W)',\qquad \forall t\in\RR.$$
 \end{proposition}

\begin{proof}
As $\Poi$ acts transitively on wedges, there is no loss of generality if we fix $W=W_3$ and consider  $G_W=G_3\subset\tPoi$. As $\Lambda_3(t)W_3=W_3$ for any $t\in\RR$ then $U(\lambda_3(t))H(W_3)=H(W_3)$ and, by Lemma \ref{lem:sym}, $\Delta_{H(W_3)}$ commutes with $U(\lambda_3(t))$. In particular, $t\mapsto Z_{H(W_3)}(t)$ defines a unitary one-parameter group.

By positivity of the energy and Theorem \ref{Borch}, $\Delta^{-it}_{H(W_3)}$ has the same commutation relations as boosts $U(\lambda_3(2\pi t))$ w.r.t. translations. Indeed, translations in $x_1$ and $x_2$ directions commute with  $\Delta_{H(W_3)}^{it}$ since they fix $H(W_3)$, and with $U(\lambda_3(t))$. Translations along directions $\mathrm v_+=(1,0,0,1)$ and $\mathrm v_-=(-1,0,0,1)$  have, respectively, positive and negative generators and $U(t)H(W_3)\subset H(W_3)$ for any $t>0$. Then, by Theorem \ref{Borch} $$\Delta_{H(W_3)}^{is}U_\pm(t)\Delta_{H(W_3)}^{-is}= U_\pm(e^{\mp2\pi s}t)\qquad s,t\in\RR,$$ as well as $$U(\lambda_{W_3}(-2\pi t)) U_\pm(t) U(\lambda_{W_3}(2\pi t))= U_\pm(e^{\mp2\pi s}t),\qquad s,t\in\RR$$ where $U_\pm(t)=U(t\cdot\mathrm v_\pm)$. Translations along $x_1,\,x_2,\, \mathrm v_+$ and $\mathrm v_-$ generate $\RR^4$ translations and as a consequence $Z_{H(W_3)}\in U(\T)'$.

Any element  $g\in G_3^0$ fixes the standard subspace $H(W_3)$, hence by Lemma \ref{lem:sym}, $U(g)$ commutes with the modular operator $\Delta_{H(W_3)}$. Furthermore, as $g$ fixes $W_3$, then $U(g)$ also commutes with $U(\lambda_3)$. We  conclude that $Z_{H(W_3)}\in U(G_3)'$.
\end{proof}

Note that $G_3^0=\langle r_3,\lambda_3, r(2\pi)\rangle$, where $r(2\pi)$ is the $2\pi$ rotation.

\begin{proposition}\label{prop:then} Let $(U,H)$ be a Poincar\'e covariant net of standard subspaces.
Let $W\in\W$, and $r_W\in\tPoi$ be such that $\Lambda(r_W)W=W'$. Assume that $Z_{H(W)}$ commutes with $U(r_W)$, then the  B-W and Duality properties hold.
\end{proposition}
\begin{proof} The map $$\R\ni t\mapsto Z_{H(W)}(t) $$ is a unitary,  one-parameter,  s.o.-continuous group by Proposition \ref{prop:sym}.
Now by hypothesis and covariance
$$Z_{H(W')}(t)=U(r_W)Z_{H(W)}(t)U(r_W)^*=Z_{H(W)}(t)$$  where
$$Z_{H(W')}(t)=U(\lambda_W(-2\pi t))\Delta_{H(W')}^{it}.$$
We find that
$$Z_{H(W)}(2t)=Z_{H(W)}(t)Z_{H(W)}(t)=Z_{H(W)}(t)Z_{H(W')}(t)=\Delta_{H(W)}^{it}\Delta_{H(W')}^{it}$$
and since $Z_{H(W)}(2t)$ is an automorphism of $H(W)$ and 
\begin{align*}\Delta_{H(W)}^{-it}Z_{H(W)}(2t)H(W)&=\Delta_{H(W')}^{it}H(W)\Leftrightarrow\\
 \Delta_{H(W)}^{-it}H(W)&=\Delta_{H(W')}^{it}H(W) \Leftrightarrow\\
H(W)&=\Delta_{H(W')}^{it}H(W)\qquad\forall t\in\RR.\end{align*}
By locality $H(W)\subset (B\,H(W'))'$, Lemma \ref{lem:sym} and Proposition \ref{prop:B}, we have$$\Delta^{it}_{(BH(W'))'}H(W)=\Delta^{-it}_{BH(W')}H(W)=\Delta^{-it}_{H(W')}H(W)=H(W)$$ and by Lemma \ref{inc} we conclude wedge duality, $$H(W)= (B\,H(W'))'.$$ Furthermore, by the last condition for any $W\in\W$ then 
$$\Delta_{H(W)}=\Delta_{H(W')}^{-1}$$
and the B-W property follows since 
\begin{align*}Z_{H({W})}(t)&=U(r_W)Z_{H(W)}(t)U(r_W)^*\\&=U(\lambda_W(-2\pi t))\Delta_{H(W)}^{-it}=Z_{H(W)}(-t)
\end{align*} hence $Z_{H(W)}(t)=\textbf{1}$. 
\end{proof}

Now we state the properties we are interested in.
\begin{definition}\label{def:mod}
We shall say that a unitary, positive energy representation is \textbf{modular} if for any $U$-covariant net of standard subspaces $H$, namely any couple $(U,H)$, then the B-W property holds.

Let $W\in\W$. A unitary, positive energy $\tPoi$-representation  $U$ satisfies the {\bf modularity condition} if for an element $r_W\in\tPoi$ such that $\Lambda(r_W)W=W'$  we have that
\begin{equation}\tag{MC}\label{eq:cond}{U(r_W)\in U(G_W)''}.\end{equation}
\end{definition}
Note that \eqref{eq:cond} does neither depend on the choice of $r_W$, nor of $W$. Indeed if $\tilde r_W\in\tPoi$ is another element such that $\Lambda(\tilde r_W) W=W'$ then $ r_W\tilde r_W\in G_W$ and if \eqref{eq:cond} holds for $U(r_W)$, then it holds for $U(\tilde r_W)$. We conclude \eqref{eq:cond} for any other wedge by the transitivity of the $\Poi$-action on wedge regions.

Now, we prove that the representations satisfying $\eqref{eq:cond}$ are modular.
\begin{theorem}\label{thm:starcor}
Let $U$ be a positive energy unitary representation of the Poincar\'e group $\tPoi$.  If the condition \eqref{eq:cond} holds on $U$, then  any local $U$-covariant net of standard subspaces,  namely any pair $(U,H)$, satisfies the B-W  and the Duality properties.
\end{theorem}

\begin{proof}
Let $(U,H)$ be a $U$-covariant net of standard subspaces, then $Z_{H(W)}\in U(G_W)'$ by Proposition \ref{prop:sym}. Then by assumptions $Z_{H(W)}$ commutes with $U(r_W)$ where $\Lambda(r_W)W=W'$, then we conclude the thesis by Proposition \ref{prop:then}
\end{proof}

Let $U$ be a representation of $\tPoi$ acting on a standard subspace net $\cW\ni W\mapsto H(W)\subset \H$. Assume that $J_{geo,W}$ is an anti-unitary operator extending $U$ to a representation of $\tilde{\cP_+}$ through $W$-reflection and  assume that modular covariance holds. As in \cite{GL}, the algebraic $J_{H(W)}$ implements the wedge $W$ reflection and, up to a $\tPoi$ element,  the PCT operator (the proof in \cite{GL} can be straightforwardly adapted in the standard subspace net case).  In this setting, let $K_W$ be the $W$-boost generator on $\H$, the following operators \begin{equation} \label{eq:algpt}S_{geo,W}\dot=J_{geo,W}e^{-\pi K_W},\qquad S_{alg,W}\dot=J_{H(W)}\Delta_{H(W)}^{1/2},\end{equation}  can be called {\bf geometric} and {\bf algebraic Tomita operators}.
\begin{corollary}
With the assumptions of Proposition \ref{prop:then}, 
$$S_{geo,W}=CS_{alg,W},\qquad\forall W\in\cW$$ where $C\in U(\tPoi)'$\end{corollary}
\begin{proof}  By duality, $J_{geo,W}$ and $J_{H(W)}$ both  implement anti-unitarily $U(j_W)$, then $J_{geo}J_{H(W)}\in U(\cP_+^\uparrow)'$. By the B-W property, $e^{-\pi K_W}=\Delta^{1/2}_W$, and we conclude.\end{proof}

The B-W property  and the $\tilde \P_+$ covariance do not imply that there is a unique net undergoing the $U$-action. The conjugation operator can differ from the geometric conjugation by a unitary in  $U(\cP_+^\uparrow)'$. For instance, given an irreducible, (anti-)unitary $\mathcal P_+$-representation $U$ where $U(j_W)$ implements the $W$-reflection, we have  two $U$-covariant nets, according to the couples $\{(U(j_W),e^{-2\pi K_W})\}_{W\in\W}$ and $\{(-U(j_W),e^{-2\pi K_W})\}_{W\in\W}$ defining the wedge subspaces. If we just require $\tPoi$-covariance through a representation $U$, then for any modulus one complex number $\lambda$, the couples $\{(\lambda U(j_W),e^{-2\pi K_W})\}_{W\in\W}$ define $U$-covariant standard subspace nets.

\subsection*{Direct sums}

The modularity property easily extends to direct integrals and multiples of representations as the following proposition shows.
\begin{proposition}\label{prop:P} Let $U$  and $\{U_x\}_{x\in X}$ be unitary positive energy representations of $\tPoi$  satisfying \eqref{eq:cond}.

Let $\K$ be an Hilbert space, then \eqref{eq:cond} holds for $U\otimes 1_\K\in\B(\H\otimes \K)$.

Let $(X,\mu)$ be a standard measure space.
Assume that  $U_x|_{G_W}$ and $U_y|_{G_W}$ are disjoint for $\mu$-a.e. $ x\neq y$ in $X$. Then $U=\int_X U_x d\mu(x)$ satisfies \eqref{eq:cond}. \end{proposition}
\begin{proof} We can assume $W=W_3$.
Let  $\cU\dot=U\otimes 1_K$, since $\cU(G_3)'=U(G_3)'\otimes B(\K)$ it follows that $\cU(r_1(\pi))=U(r_1(\pi))\otimes1$ commutes with $\cU(G_3)'$, hence $\cU(r_1(\pi))\in \cU(G_3)''$.

For the second statement, let
$$U(a,\Lambda)=\int_{X}^\bigoplus U_x(a,\Lambda)d\mu(x)\quad\text{acting on }\quad\H=\int_{X}^\bigoplus\H_xd\mu(x).$$
Then, by disjointness,  $$U(G_3)''=\int_{X}^\oplus U_x(G_3)''d\mu(x)$$ and $U(r_1)=\int_{X}^\oplus U(r_1(\pi))d\mu(x)$ we have that $U(r_1(\pi))\in U(G_3)''$\end{proof}

\section{The scalar case}

In this section we are going to show that \eqref{eq:cond} holds for scalar representations, namely  massive or massless $0$-spin/helicity representations. 
The scalar representations have the following form 
$$(U_{m,0}(a,A)\phi)(p)=e^{iap}\phi(\Lambda(A)^{-1}p),\qquad(a,A)\in\RR^{1+3}\rtimes\tilde \cL_+^\uparrow=\tPoi,$$
where $$\phi\in\H_{m,0}\dot= L^2(\Omega_m,\delta(p^2-m^2)\theta(p_0)d^4p),$$ and $\Omega_m=\{p=(p_0,{\bf p})\in\RR^{1+3}:\,p^2=p_0^2-{\bf p}^2=m^2\}$ with $m\geq0$. 
We recall that $U_{m,0}$ factors through $\Poi$.

Any momentum $p\in\Omega_m$ is a point in the dual group of $\T$ i.e. a character. We recall that $\Poi$ acts on $\Omega_m$-characters as dual action of the adjoint action of $\Poi$ on $\T$.  Clearly, $\T$  acts trivially and $\mathcal L_+^\uparrow$ acts geometrically on $\Omega_m$, i.e. $(a,\Lambda)\cdot p=\Lambda^{-1}p$ with $(a,\Lambda)\in\RR^{1+3}\rtimes \cL_+^\uparrow=\Poi$.

We start  with the following remark.
\begin{remark} \label{rmk:orb} Fix $p=(p_0,p_1,p_2,p_3)\in\Omega_m$, $m> 0$,  $$R_1(\pi)p=(p_0,p_1,-p_2,-p_3)$$ can be obtained as a composition of a $\Lambda_3$-boost of parameter $t_p$ and a $R_3$-rotation of parameter $\theta_p$ as
\begin{equation}\label{'''}
\begin{aligned}\Lambda_3(t_p)R_3(\theta_p)(p_0,p_1,p_2,p_3)&=\Lambda_3(t_p)(p_0,p_1,-p_2,p_3)\\&=(p_0,p_1,-p_2,-p_3).\end{aligned}\end{equation} Clearly $t_p$ and $\theta_p$ depend on $p$. By \eqref{'''}, we deduce that  $G_3^0$ orbits on $\Omega_m$ are not changed by $R_1(\pi)$. With $m=0$ an analogue argument holds for all the orbits except $\{(p_0,0,0,p_0), p_0\geq0\}$ and $\{(p_0,0,0,-p_0), p_0\geq0\}$, i.e. there is no $g\in G_3$, such that $g(p_0,0,0,-p_0)=(p_0,0,0,p_0)$. On the other hand these orbits have null measure with the restriction of the Lebesgue measure to $\partial V_+$. This remark holds in $\RR^{1+s}$ with $s>2$. 
\end{remark}

\begin{lemma}\label{lem:orb}
Let $f\in L^\infty(\Omega_m)$ such that for every $g\in G_3^0$, $f(p)=f(gp)$ for a.e. $p\in\Omega_m$. Then $f(p)=h(p)$ for a.e. $p\in\Omega_m$ where $h\in L^\infty(\Omega_m)$ is  constant  on any $\overline{\{g \,p\}}_{g\in G_3}$ orbit. 
\end{lemma} 
\begin{proof}
Any point $p=(p_0,p_1,p_2,p_3)\in\Omega_m$ (except the massless null measure sets $\{(p_0,0,0,p_0), p_0\geq0\}$ and $\{(p_0,0,0,-p_0), p_0\geq0\}$), can be identified with a radius $r=p_1^2+p_2^2$, an angle $\theta\in[0,2\pi]$ such that $(p_1,p_2)=r(\cos\theta,\sin\theta)$  and a parameter $t\in\RR$ such that if $(p_0,p_3)=\sqrt{r^2+m^2}(\cosh t ,\sinh t)$.
In particular $G_3^0$ orbits  $\sigma_r$  are labelled by $r\in\RR^+$. On each orbit (with its invariant $G_3^0$ measure) the only positive measure set which is preserved by the $G_3^0$ action is the full orbit (up to a set of measure zero). This fact ensures that $G_3^0$ invariant functions have to be almost everywhere constant.
Up to unitary equivalence, we can decompose the Hilbert space as $\int_{\RR^+}L^2(\sigma_r,d\mu_r)dr$ where $d\mu_r$ is the $G_3$ invariant measure on $\sigma_r$. 
$G_3$ representations on different orbits are inequivalent, then $U(G_3)'=\int_{\RR^+}(f(r)\cdot1)dr$ and we conclude.
\end{proof}
\begin{proposition}\label{prop:scalar} Let $U$ be a unitary, positive energy, irreducible scalar representation of the Poicar\'e group. Then $U$ satisfies the modularity condition \eqref{eq:cond}.
\end{proposition}
\begin{proof}It is enough to consider $W_3$. Let $Z\in U(G_3)'$. Since $U$ is a scalar representation then the translation algebra $\T''=\{U(a):a\in\RR^4\}''$ is a MASA  and $\T''=L^\infty(\Omega_m)$. 
Indeed, the translation unitaries  $(U(a)\phi)(p)=e^{iap}\phi(p)$ are multiplication operators and generate $L^\infty(\Omega_m)$ ultra-weakly. In particular $Z\in\T'=\T''$, hence it is a multiplication operator $Z=M_f$ by $f\in L^\infty(\Omega)$. 
Furthermore, $Z$ has to commute with $U(\Lambda_3)$ and $U(R_3)$, hence  $\forall t\in\RR$ and $\forall\theta \in [0,2\pi]$
$$U(\Lambda_3(t))ZU(\Lambda_3(t))^*=Z\quad\Leftrightarrow \quad f(\Lambda_3(t)^{-1}p)=f(p),\qquad \text{a.e.}\,p\in\Omega_m$$ 
and 
$$U(R_3(\theta))ZU(R_3(\theta))^*=Z\quad \Leftrightarrow f(R_3(\theta)^{-1}p)=f(p)\qquad \text{a.e.}\,p\in\Omega_m.$$ We can assume $f(p)$ to be constant on any $\Lambda_3$ and  $R_3$ orbit by Lemma \ref{lem:orb}. 

Now observe that any momentum on the hyperboloid $p=(p_0,p_1,p_2,p_3)$ can be connected to  the $R_1(\pi)p=(p_0,p_1,-p_2,-p_3)$ through a $R_3$-rotation and a $\Lambda_3$-boost as in Remark \ref{rmk:orb}. It follows that $f(p)=f(R_1(-\pi)p)$, for every $p\in\Omega_m$. As a consequence $$U(R_1(\pi))ZU(R_1(\pi))^*=Z\quad\text{as}\quad f(R_1(\pi)^{-1}p)=f(p),\qquad \forall p\in\Omega_m$$ and we conclude.
\end{proof}

Now, we can state the theorem.
\begin{theorem}\label{starint}Let $U=\int_{[0,+\infty)}{U_m}d\mu(m)$ where $\{U_m\}$ are (finite or infinite) multiples of the scalar representation of mass $m$,  then $U$ satisfies \eqref{eq:cond}. In particular if $(U,H)$ is a $U$-covariant net of standard subspaces, then the Duality and the B-W properties hold.\end{theorem} 

\begin{proof} Unitary representations of $\Poi$ with different masses are disjoint, and they have disjoint restrictions to $G_3$ subgroup. The thesis becomes a consequence of Propositions  \ref{prop:P},  \ref{prop:scalar} and Theorem \ref{thm:starcor} .
\end{proof}

\begin{remark}\label{rmk:hel} Proposition \ref{prop:scalar} straightforwardly holds also for irreducible massless finite helicity representations as they are induced from a one-dimensional representation of the little group.
As a  consequence, an irreducible nonzero helicity representation $U$ cannot  act covariantly on a net of standard subspaces on wedges $H$. Indeed, the B-W property must hold as in Proposition \ref{prop:scalar} and $J_{W_3}$, the modular conjugation of $H(W_3)$, would implement the $j_3$ reflection on $U$ (cf. \cite{GL}). In particular, the PT operator  defined by $\Theta=J_{W_3}U(R_3(\pi))$  extends $U$ to an (anti-)unitary representation $\hat U$ of $\tilde\P_+$ and acts  covariantly  on $H$. On the other hand nonzero helicity representations are not induced by a self-conjugate representation of the little group and do not  extend to  anti-unitary representations of $\tilde{\P}_+$ \cite{var}. This shows a contradiction. 
\end{remark}
\begin{remark}\label{rmk:dim}
Consider the $\RR^{1+s}$ spacetime with $s\geq3$ and let $U$ be a scalar (unitary, positive energy) representation of $\Poi$. The one-parameter group $t\mapsto Z_{H(W_1)}(t)$ given in \eqref{eq:Z},   is generated by the multiplication operator by a real function of the form $f(p_2^2+...+p_s^2)$.  For each value of the radius $r=p_2^2+...+p_s^2$  there is a unique  $G_1$-orbit on $\Omega_m$  which is fixed by any transformation $R\in\Poi$ such that $RW_1=W_1'$, for instance $R_2(\pi)$.  In particular, the analysis of this section extends to any Minkowski spacetime $\RR^{1+s}$ with $s\geq3$. It fails in $2+1$ spacetime dimensions as $R_2(\pi)$ does not preserve $G_1$-orbits.

\end{remark}

\section{Massive counter-examples} \label{sect:infcoun}

Borchers, in \cite{bor}, showed that a unitary, positive energy Poincar\'e representation acting covariantly on a modular covariant von Neumann algebra net $\A$ in the vacuum sector, can only  differ from the modular representation  by a unitary representation of the Lorentz group. 
\begin{theorem}\label{bor}{\rm \cite{bor}.} Let $\A$ be a local quantum field theory von Neumann algebra net in the vacuum sector undergoing two different representations of the Poincar\'e group. Let $U_0$ be the representation implemented by wedge modular operators and $U_1$ be the second representation. Then there exists a unitary representation of the Lorentz Group $G(\Lambda)$ defined
$$G(\Lambda)=U_1(a,\Lambda)U_{0}(a,\Lambda)^*.$$
Furthermore, $G(\Lambda)$ commutes with $U_0(a,\Lambda')$ for all $a\in\RR^{1+3}$, $\Lambda,\,\Lambda'\in\cL_+^\uparrow$ and the $G(\Lambda)$ adjoint action on $\A$ implements automorphisms of local algebras, i.e. maps any local algebra into itself.
\end{theorem}

With this hint it is possible to build up Poincar\'e covariant nets picturing the above situation: {\it modular covariance without the B-W property}. We are going to make explicit computations on this kind of counterexamples in order to understand what may prevent the B-W property. Here, we study the massive case. The massless case can be found in \cite{LMR}.

Consider $U_{m,s}$, the  $m$-mass, $s$-spin, unitary, irreducible representation of the  Poincar\'e group $\tPoi$ and $H:W\mapsto H(W)$ its canonical net of standard subspaces. 
Let $p\in\Omega_m$, $A_p\dot=\sqrt{{\utilde p}/{m}}$, where $\utilde p=p_0\cdot \eins+\sum_{i=1}^3p_i\sigma_i$ is the $\SLC$ element  implementing the boost sending the point $q_m=(m,0,0,0)$ to $p$. An explicit form of $U_{m,s}$ is the following
$$(U_{m,s}(a,A)\phi)(p)=e^{iap}\D^s(A_p^{-1}AA_{\Lambda^{-1}p})\phi(\Lambda(A)^{-1}p),$$
where $\D^s$ is the $s$-spin representation of $\SU(2)$ on the $2s+1$ dimensional Hilbert space $\mathfrak{h}_s$ and 
\begin{align*}
\phi\in\H_{m,s}&\dot=\mathfrak{h}_s\otimes L^2(\Omega_m,\delta(p^2-m^2)\theta(p_0)d^4p)\\&=\CC^{2s+1}\otimes L^2(\Omega_m,\delta(p^2-m^2)\theta(p_0)d^4p).
\end{align*}

Let $V$ be a real unitary, nontrivial,  $\SLC$-representation  on an Hilbert space $\K$, i.e. there exists an anti-unitary involution $J$ on the Hilbert space $\K$, commuting with $V$ such that the  real vector space $K\subset\K$ of $J$-fixed vectors,  is a standard subspace and $$V(\SLC)K=K.$$ In particular $JK=K$ and $ \Delta_K=1$. 

We can define the following net of standard subspaces, $$K\otimes H:\W\ni W\rightarrow K\otimes H(W) \subset\K\otimes \mathcal{H}.$$  
We can see  two Poincar\'e representations acting covariantly on $K\otimes H$: $$U_I(a,A)\equiv1\otimes U_{m,s}(a,A)\qquad A\in \SLC,\,a\in\RR^{1+3}$$ and $$U_V(a,A)\equiv V(A)\otimes U_{m,s}(a,A)\qquad A\in \SLC,\,a\in\RR^{1+3}.$$ $U_I$ is implemented by $K\otimes H$ modular operators and the B-W property holds w.r.t. $U_I$ (cf. Lemma 2.6 in \cite{LMR}).   Note that for $s=0$,  $U_I$ satisfies the condition (\ref{eq:cond}), hence the B-W property  also by Theorem \ref{thm:starcor}. 

$U_V $ decomposes in a direct sum of infinitely many inequivalent representations of mass $m$, i.e. infinitely many spins appear. Indeed, consider the Lorentz transformation $A_p=\sqrt{\utilde p/m}$ and the unitary operator on $\H\otimes\K$ 
\begin{equation}\label{intert}W:\K\otimes\H\ni\left(p\mapsto\phi(p)\right)\longmapsto \left(p\mapsto (V(A_p^{-1})\otimes1_{\CC^{2s+1}})\phi(p)\right)\in\K\otimes\H\end{equation}
as $\K\otimes\H=L^2(\Omega_m,\K\otimes\CC^{2s+1},\delta(p^2-m^2)\theta(p_0)d^4p).$
We can define a unitarily equivalent representation $U'_V=WU_VW^*$ as follows:
$$(U_V'(a,A)\phi)(p)=e^{ia\cdot p}(V(A_p^{-1}AA_{\Lambda(A)^{-1}p})\otimes \D^s(A_p^{-1}AA_{\Lambda(A)^{-1}p}))\phi(\Lambda(A^{-1}p)).$$
It is easy to see that $A_p^{-1}AA_{\Lambda(A)^{-1}p}\in\mathrm{Stab}(m,0,0,0)=\SU(2)$.

As we are interested in the disintegration of $U_V$ it is not restrictive to assume that $V$ is irreducible. Unitary irreducible representations  of $\SLC$, denoted by $V_{\rho,n}$, are labelled by pairs of numbers $(\rho,n)$ such that  $\rho\in\RR$ and $n\in\ZZ_+$. The restriction of $V_{\rho,n}$ to $\SU(2)$ decomposes in $\bigoplus_{s=n/2}^{+\infty} \D^s$ (see, for example,  \cite{Ru,CH}). 
Thus, if  $V=V_{\rho,n}$ is an irreducible $\SLC$ representation then
\begin{equation}\label{eq:decomp} 
U_V\simeq\bigoplus_{i=\frac{n}{2}}^\infty\,\bigoplus_{j=|s-i|}^{s+i}U_{m,j}\end{equation} 
since $U'_V$, hence $U_V$, decomposes according to the $(V\otimes \D^s)|_{\SU(2)}$ decomposition into irreducible representations.

Note that  any  representation class appears with finite multiplicity and  $U_V$ does not satisfy the condition \eqref{eq:cond}. Furthermore, we can conclude that having modular covariance without the B-W property  requires to the ``wrong" representation the presence of an infinite family of inequivalent Poincar\'e representations possibly with finite multiplicity.

$U_I$ cannot act covariantly on the $U_V$-canonical net $H_V$. We can see it explicitly. Let $W_1$ be the wedge in $x_1$ direction and $W=gW_1$ where $g\in\cL_+^\uparrow$ and $W_1\neq W$. If $U_I(g)H_V(W_1)=H_V(W)$ then $U_I(g)\Delta_{H_V(W_1)}^{-it}U_I(g)^*=\Delta_{H_V(W)}^{-it}$ where $\Delta_{H_V(W_1)}$ and $\Delta_{H_V(W)}$ are the modular operator of $H_V(W_1)$ and $H_V(W)$, respectively. As
\begin{align*}
U_I(g)&\Delta_{H_V(W_1)}^{-it}U_I(g)^*=\\
&=(1\otimes U_{m,s}(g))(V(\Lambda_{W_1}(2\pi t))\otimes U_{m,s}(\Lambda_{W_1}(2\pi t)))(1\otimes U_{m,s}(g))^*\\
& 
=V(\Lambda_{W_1}(2\pi t))\otimes U_{m,s}(g\Lambda_{W_1}(2\pi t))g^{-1})\end{align*} 
but $\Delta_{H_V(W)}^{-it}=V(g\Lambda_{W_1}(2\pi t)g^{-1})\otimes U_{m,s}(g\Lambda_{W_1}(2\pi t))g^{-1})$, we get the contradiction unless $V$ is trivial.


The modular covariance and the identification of the geometric and the algebraic PCT operators  imply the uniqueness of the covariant representation. 
\begin{proposition}
With the assumptions of Theorem \ref{bor} assume that the algebraic PCT operator defined in  \eqref{eq:algpt} implements the $U_1$-PCT operator too.
Then $U_1=U_0$.
\end{proposition}
\begin{proof} With the notations of Theorem \ref{bor}, we consider the representation $U_1(a,\Lambda) = G(\Lambda)U_0(a,\Lambda)$ of $\tPoi$  on $\H$. $G$ fixes any subspace $H(W)=\overline{\A(W)\Omega}$ and in particular $J_{H(W)}G(\Lambda)J_{H(W)}=G(\Lambda)$, $\forall \Lambda\in\cL_+^\uparrow $. It follows that $J_{H(W)}$ has the correct commutation relations with  $U_1$ if and only if $G=1$.\end{proof}

We have   seen that  given a Poincar\'e covariant net with the B-W property $(U,H)$ we can find a second covariant Poincar\'e representation $\tilde U$  if the commutant $U(\tPoi)'$ is large enough, i.e. when the Poincar\'e representation implemented by the modular operators has infinite multiplicity. Furthermore, there are no conditions on the spin of $\tilde U$. In particular, counterexamples can produce wrong relations between spin and statistics. For instance assuming $U_I(r(2\pi))={\bf1}$ and $V(r(2\pi))=-{\bf1}$ where $r(2\pi)$ is the $2\pi$-rotation. The \eqref{eq:cond} condition does not hold for $U_V$ and a wrong spin-statistics relation may arise whenever the B-W property fails.

As we said in the introduction known counterexamples to modular covariance are very artificial and not all the basic relativistic and quantum assumptions are respected. It is interesting to look for natural counterexamples to modular covariance, in the class of representations excluded by this discussion, if they exist. In \cite{DHRII} it is shown that assuming finite one-particle degeneracy, then the Spin-Statistics Theorem holds. This accords with the analysis obtained by Mund in  \cite{M}. We expect that an algebraic proof for the B-W property can be established, without assuming finite multiplicity of sub-representations. 

\section{Outlook}\label{sect:CO}

Now we come to an outlook on  the relation between the split and the B-W properties.
\begin{definition}\label{def:split}\cite{lodo}
Let $(\N\subset \M,\Omega)$ be a \textit{standard inclusion} of von Neumann algebras, i.e.\! $\Omega$ is a cyclic and separating vector for $N,\, M$ and $N'\cap M$.

A standard inclusion $(\N\subset \M,\Omega)$ is \textit{split} if there exists a type I factor $\F$ such that $\N\subset \F \subset \M$.

A Poincar\'e covariant net $(\A,U,\Omega)$ satisfies the \textit{split property} if  the von Neumann algebra inclusion $(\A(O_1)\subset\A(O_2),\Omega)$ is split, for every compact inclusion of bounded causally closed regions $O_1\Subset O_2$.
\end{definition}
In a natural way one can define the split property for an inclusion of subspaces by second quantization: {\it let $K\subset H \subset \H$ be an inclusion of standard subspaces of an Hilbert space $\H$ such that $K'\cap H$ is standard, then the inclusion $K\subset H$ is split if its second quantization inclusion $R(K)\subset R(H)$ is standard split w.r.t. the vacuum vector $\Omega\in e^\H$}.  The second quantization respects the lattice structure (cf. \cite{A}). Here, we just deal with the Bosonic second quantization (see, for example, \cite{LRT}).

We want a coherent first quantization version of the split property. We assume $H$ and $K$ to be {\bf factor} subspaces, i.e. $H\cap H'=\{0\}=K\cap K'$. We need the following theorem.
\begin{theorem}\label{teo:gf}{\rm \cite{figu}.} Let $H$ be a standard subspace and $R(H)$ be its second quantization.
\begin{itemize}
\item Second quantization factors are type I if and only if $\Delta_H|_{[0,1]}$ is a trace class operator where $\Delta_{H}|_{[0,1]}$ is the restriction of the $H$-modular operator $\Delta_H$ to the spectral subspace relative to the interval ${[0,1]}$
\item Second quantization factors which are not type I are type III. 
\end{itemize}
\end{theorem}
The canonical intermediate factor,  for a standard split inclusion of von Neumann algebras $N\subset M$ is
$$\F=N\vee JNJ=M\cap JMJ$$ where $J$ is the modular conjugation associated with the relative commutant algebra $(N'\cap M,\Omega)$, cf. \cite{lodo}.  The standard inclusion $(N\subset M,\Omega)$ is split iff $\F$ is type I.
A canonical intermediate subspace can be analogously defined for standard split subspace inclusions $K\subset H$:
$$F=\overline{K+JK}=H\cap JH.$$ Second quantization of modular operators were computed in \cite{EO,LRT,F}.

Consider the following proposition:
\begin{proposition}\label{prop:wMB}
Let $K\subset H$ be an inclusion of standard subspaces, such that $K'\cap H$ is standard. Let $J$ be the modular conjugation of the symplectic relative complement $K'\cap H$ and $\Delta=\Delta_F$ be  the modular operator of the intermediate subspace $F=\overline{K+JK}$. Assume that $F$ is a factor. Then the following statements are equivalent.
\begin{itemize}
\item[1.] $ R(K)\subset R(H)$ is a split inclusion.
\item[2.] the operator $\Delta |_{[0,1]}$ is trace class.
\end{itemize}
\end{proposition}
\begin{proof}
($1.\Rightarrow 2.$) If the inclusion is split, then the intermediate canonical factor $R(K)\vee R(JK)$ is the second quantization of $\overline {K+ JK}$. In particular by Theorem \ref{teo:gf} the thesis holds.\\
($2.\Rightarrow 1.$) As $\Delta|_{[0,1]} $ is trace class, then the second quantization of $F$ is an intermediate type I factor between $R(K)$ and $R(H)$ by Theorem \ref{teo:gf}. Then the split property holds for the inclusion $R(H)\subset R(K)$.
\end{proof}
Let $U$ be an (anti-)unitary representation of $\cP_+$. We shall say  that $U$ is {\bf split} if the canonical net associated with $U$ satisfies the split property on bounded causally closed regions (defined through equation \eqref{HO}), i.e. the inclusion
$$H(O_1)\subset H(O_2)$$  is split for every $O_1\Subset O_2$ as above.

Scalar free fields satisfy Haag duality, thus \eqref{HO} holds (see, for example, \cite{O}). Furthermore, any scalar irreducible  representation is split (cf. \cite{BW,BDL}).
The following theorem is the first quantization analogue of Theorem 10.2 in \cite{lodo}.
\begin{theorem}\label{344}
Let $U$ be an (anti-)unitary representation of $\cP_+$, direct integral of scalar representations. If $U$ has the split property then $U=\int_{0}^{+\infty}U_m d\mu(m)$ where  $\mu$ is purely atomic on isolated points and for each mass there can only be a finite multiple of $U_{m,0}$.\end{theorem}
\begin{proof}
As the B-W property holds, the net disintegrates according to the representation:
$$H=\int_0^{+\infty} H_m d\mu(m)$$where $H_m$ is the canonical net associated with $U_m.$\\
Fix a couple of bounded and causally closed regions $O_1\Subset O_2\subset\RR^{1+3}$. Irreducible components satisfy the split property. In particular by Proposition \ref{prop:wMB}, the restriction of the modular operator of the intermediate standard subspace $F$, defined as 
\begin{align*}F=\overline{H(O_1)+JH(O_1)}&=\int_0^{+\infty}\overline{ H_m(O_1)+ J_mH_m(O_1)_m} d\mu(m)\\&=\int_0^{+\infty}F_m\,d\mu(m),
\end{align*}
 to the spectral subset $[0,1]$,  has to be trace class. It follows that  $\mu$ has to be purely atomic and  for each isolated mass (no finite accumulation point) there can only be a finite multiple of the scalar representation.
\end{proof}
\begin{corollary} Let $(U,H)$ be a Poincar\'e covariant net of standard subspaces and $U$ be a $\Poi$-split representation. Assume that $U$ is a direct integral of scalar representations. Then the B-W and the duality properties hold.
\end{corollary}
\begin{proof}
By Theorem \ref{344} the disintegration of the covariant Poincar\'e representation is purely atomic on masses, concentrated on isolated points and for each mass there can only be a finite multiple of the scalar representation. The disintegration satisfies the condition \eqref{eq:cond}  and the thesis follows by Theorem \ref{thm:starcor}.
\end{proof}

It remains an old interesting challenge to build up a more general bridge between the Split and the B-W properties. We expect this analysis to be generalized to finite multiples of spinorial representations.

\section*{Acknowledgements} I thank Roberto Longo for suggesting me the problem and useful comments and M. Bischoff for discussing Proposition \ref{prop:wMB}. I also thank Sergio Ciamprone and Yoh Tanimoto for comments.


\begin{thebibliography}{00}


\bibitem{A} 
{H.\! Araki}, 
{A lattice of von Neumann algebras associated with the Quantum Theory of a free Bose field}, \textit{J. Math.\ Phys.}\ \textbf{4} (1963), 1343--1362.


\bibitem{BW1} 
{ J.\! J.\! Bisognano, E.H. Wichmann}, 
{On the Duality condition for quantum fields}, {\it J. Math. Phys.} \textbf{17} (1976), 303--321.

\bibitem{bw2} 
{ J. \! J.\!  Bisognano, E.\! H.\! Wichmann}, 
{On the Duality condition for Hermitian scalar field}, {\it J. Math. Phys.} \textbf{16} (1975), 985--1007.

\bibitem{BGL} { R.\! Brunetti, D.\! Guido, R.\! Longo}, 
{ Modular localization and Wigner particles}, 
{\it Rev. Math. Phys.} {\bf 14}, N. 7 \& 8 (2002), 759--786. 

\bibitem{BGLco}{ R.\! Brunetti, D.\! Guido, R.\! Longo}, {Group cohomology, modular theory and space-time symmetries,}{\it Rev. Math. Phys.}  \textbf{7}, (1994), 57--71.

\bibitem {bor} { H.\! J.\! Borchers}, { On {P}oincar\'e transformations and the modular group of the algebra associated with a wedge}, {\it Lett. Math. Phys.} {\textbf{46}}, {no. 4}, (1998), 295--301


\bibitem{B} { H.\! J.\! Borchers}, 
{ The CPT Theorem in two-dimensional theories of local observables}, 
{\it Comm. Math. Phys.} {\bf 143} (1992), 315--332. 

\bibitem{BDFS}{ D.\! Buchholz,  O.\! Dreyer, M.\! Florig, S.\! J.\! Summers}, { Geometric modular action and spacetime symmetry groups}, {\it Rev. Math. Phys.}, {\bf 12},  no. 4, (2000), 475--560.

\bibitem{BDL}{ D.\! Buchholz, C.\! D'antoni, R.\! Longo}, { Nuclearity and thermal states in conformal field theory},\textit{ Comm. Math. Phys.} {\bf 270}, no.1 (2006), 267-–293.

\bibitem{BE} { D. Buchholz, H. Epstein}, { Spin and statistics of quantum topological charges}, {\it Fysica}, {\bf 17}, (1985), {329--343}.

\bibitem{BW} { D.\! Buchholz, E.\! Wichmann}, { Causal independence and the energy-level density of states in local quantum field theory}, \textit{Comm. Math. Phys.} {\bf106} (1986), no. 2, 321--344. 

\bibitem{CH}{F.\! Conrady, J.\! Hnybida}, { Unitary irreducible representations of $\SLC$ in discrete and continuous $\mathrm{SU}(1,1)$ bases}, \textit{J. of Math. Phys.} \textbf{52} (2011), 012501.

\bibitem{EO} { J.\! P.\! Eckmann, K.\! Osterwalder},  
{ An application of Tomita's theory of modular Hilbert algebras: duality for free Bose Fields},
\textit{J. Funct.\ Anal.}\ {\bf 13} (1973), 1--12.


\bibitem{lodo} { S.\! Doplicher, R.\! Longo}, { Standard and Split inclusions of von Neumann algebras}, {\it Invent. Math.}, {\bf 75}, {3},  {493--536}, (1984).

\bibitem{DHRII}{ S.\! Doplicher, R.\! Haag, J.\! E.\! Roberts}, { Local observables and particle statistics. II}, \textit{Comm. Math. Phys. }{\bf35} (1974), 49--85. 

\bibitem{figu}{ F.\! Figliolini and D.\! Guido}, { On the type of second quantization factors}, {\it J. Operator Theory}, {\bf 31},  (1994), {229--252}.

\bibitem{F} { J.\! J.\! Foit}, 
{ Abstract twisted duality for quantum free Fermi fields},
\textit{Publ.\ RIMS, Kyoto Univ.}\ {\bf 19} (1983), 729--74.


\bibitem{GL} { D.\! Guido, R.\! Longo}, 
{ An algebraic spin and statistics theorem}, 
{\it Comm. Math. Phys.} {\bf 172}, (1995), 517--533. 

\bibitem{GY}{J. Gaier, J. Yngvason }, Geometric Modular Action, Wedge Duality and Lorentz Covariance are Equivalent for Generalized Free Fields, {\it J. Math. Phys.} \textbf{41}, 5910–5919 (2000)

\bibitem{LL}
{ G.\! Lechner, R.\! Longo}, {Localization in nets of standard spaces}, \textit{Comm.  Math. Phys. }\textbf{336} (2015),  27--61.

\bibitem{LRT} { P.\! Leyland, J.\! E.\! Roberts, D.\! Testard}, {Duality for quantum free fields}, unpublished manuscript, Marseille 1978. 

 \bibitem{L}{ R.\!  Longo}, 
``Lectures on Conformal Nets", preliminary lecture notes that are available at
http://www.mat.uniroma2.it/{$\sim$}longo/Lecture-Notes.html\ . 

\bibitem{LN}
{R.\! Longo},  
{ Real Hilbert subspaces, modular theory, $\SL2$ and CFT}, in:
``Von Neumann algebras in Sibiu'', pp.\ 33--91, {\it Theta Ser. Adv. Math.}, {\bf10}, Theta, Bucharest (2008). 

\bibitem{LMR}{ R.\! Longo, V.\! Morinelli, K.\! -H.\! Rehren},{ Where Infinite Spin Particles Are Localizable}, {\it Comm. in Math. Phys.}, \textbf{345} (2016),  587-–614.

\bibitem{LW}{R.\! Longo, E.\! Witten}, {An algebraic construction of boundary Quantum Field Theory,}
{\it Comm. in Math. Phys.} \textbf{303} (2011) , 213-232.

\bibitem{M} {J.\! Mund}, 
{ The Bisognano-Wichmann theorem for massive 
theories},  {\it Ann. Henri Poincar\'e} {\bf 2} (2001), 907-926. 

\bibitem{MMG}{V.\! Morinelli}, {An algebraic condition for the Bisognano-Wichmann property}, {\it Proceedings of the 14th Marcel Grossmann Meeting - MG14} , Rome, (2015), arXiv:1604.04750, (to appear).

\bibitem{NO} K.-H. Neeb, G. Olafsson, {Antiunitary representations and modular theory}, to appear in "50th Sophus Lie Seminar", Eds. K. Grabowska, J. Grabowski, A. Fialowski and K.-H. Neeb; Banach Center Publication.

\bibitem{O}{K. Osterwalder}, {Duality for Free Bose Fields}, {\it Comm. Math. Phys.} {\bf 29}, (1973), 1--14.

\bibitem{RV} { M.\! A.\! Rieffel, A.\! Van Daele}, 
{A bounded operator approach to Tomita-Takesaki theory}, 
\textit{Pacific J. Math.}\ \textbf{69} (1977), 187--221.

\bibitem{Ru}{ W.\! R\"uhl}, { The Lorentz group and harmonic analysis}, W. A. Benjamin Inc., New York (1970).

\bibitem{sch}{ B.\! Schroer, and H.\! -W.\! Wiesbrock}, {Modular theory and geometry},
{\it Rev. Math. Phys.} {12}, (2000), {\bf 1}, 139--158.

\bibitem{var}{ V.S. Varadarajan}, {Geometry of quantum theory}, {Second Edition}, {Springer-Verlag, New York}, (1985)

\bibitem{yng} { J.\! Yngvason}, { A note on essential duality}, {\it Lett. Math. Phys.}, {\bf 31}, (1994), {\bf 2}, {127--141}

\end{thebibliography}
\end{document}